\documentclass[12pt]{article}
\usepackage{geometry}
\usepackage{graphicx}
\usepackage{amsmath}
\usepackage{amssymb}
\usepackage[normalem]{ulem}
\usepackage[all]{xy}
\usepackage{paralist}

\newtheorem{theorem}{Theorem}[section]

\def\squarebox#1{\hbox to #1{\hfill\vbox to #1{\vfill}}}

\newcommand{\qed}{\hspace*{\fill}\vbox{\hrule\hbox{\vrule\squarebox{.667em}\vrule}\hrule}\smallskip}
\newenvironment{proof}{\noindent{\bf Proof:~~}}{\(\qed\)}

\usepackage{setspace}
\begin{document}

\title{I'd Rather Stay Stupid:\\ The Advantage of Having Low Utility}
\author{Lior Seeman\\
     Department of Computer Science\\
       Cornell University\\
    	lseeman@cs.cornell.edu
}
\date{}
\maketitle

\begin{abstract}
Motivated by cost of computation in game theory, we explore how changing the utilities of players (changing their complexity costs) affects the outcome of a game.
We show that even if we improve a player's utility in every action profile, his payoff in equilibrium might be lower than in the equilibrium before the change. 
We provide some conditions on games that are  sufficient to ensure this does not occur. 
We then show how this counter-intuitive phenomenon can explain real life phenomena such as free riding, and why this might cause people to give signals indicating that they are not as good as they really are.   
\end{abstract}

\section{Introduction}
We are all familiar with situations where we feel that if only that other choice was a little cheaper, we could have done so much better, or if our computer was just a bit faster, we would have been in a much better situation when facing our competitors, or if the government would have just subsidized our research, we could have been in a much better position to compete.
This paper will show that sometimes we were just wrong.

Our motivation for looking at this question comes from trying to understand if having bounded rationality is always bad for a player.
More specifically, we consider whether decreasing the cost of computation would make an agent better off, and most of the examples we give are motivated by this question. However, our ideas can easily be generalized to any change in the utility of the players.

There have been a number of attempts to take the cost of computation into account in game theory.
The first approach is to bound players' computational power by restricting them to use only finite automata.
This approach was first introduced by Neyman \cite{Neyman1985}, who showed that complexity cost can affect the outcome of a game, by showing it can explain cooperation in repeated prisoners' dilemma.
In later papers, Ben-Porath \cite{porath1993repeated} and Gilboa and Samet  \cite{gilboa1989bounded} all showed, as our intuition expects, that a bounded player has disadvantages against a much stronger player who can use a much larger automaton.
But Gilboa and Samet also showed a somewhat opposite phenomenon, which they called ``the tyranny of the weak". They showed that the bounded player might actually gain from being bounded, relative to a situation where he was unbounded, because being bounded serves as a credible ``threat". 

In this paper, we explore this counter-intuitive scenario in a slightly different framework, where, instead of being bounded, the players pay for the complexity of the strategy they use. This approach was first introduced by Rubinstein \cite{rubinstein1986finite}, where the players pay for their automaton size. Ben-Sason, Kalai and Kalai \cite{ben2007approach} studied a class of games, where, in addition to the game, each player has a cost associated with each strategy he is playing, regardless of the other players' strategies. 
These ideas were generalized by Halpern and Pass \cite{HP10}. In their framework, a player's strategy involves choosing a Turing machine, and the complexity cost of the strategy choice is a function of the machines chosen by all the players and of the input. Our paper uses a version of this framework to explore the added value for a player from having better complexity.

Our result are similar in spirits to results obtained for ``value of information", the added value in expected utility that an agent gets from having information revealed to him. Blackwell \cite{blackwell1951comparison,blackwell1953equivalent} showed that in a single-agent decision problem, the value of information is non-negative. In a multi-agent environment, the situation is more complicated, because we need to consider how the information that an agent possesses affects other agents' actions.  Hirshleifer \cite{hirshleifer1971private} and Kaimen, Tauman, and Zamir \cite{kamien1990value} showed that, in a multi-agent game, more information to a single player can result in an equilibrium in which his payoff is reduced. Neyman \cite{neyman1991positive} showed that if the other players do not know about a player's new information, that player's payoff can not be reduced in equilibrium. 

Similarly to the idea of value of information we compare games before and after a change in the utility functions of the players. As Neyman \cite{neyman1991positive} pointed out, changing a game in such a way creates a totally new game, so by comparing the utility in equilibrium before and after the change, we actually compare two different games (for example, the information of the players also changes, since the game description is different and this is part of their information), so perhaps it is not surprising that the results are not always what we might expect. Nevertheless, we also choose this approach since we feel that, although these are two different games, their most significant difference is the utility change and all other changes are caused by it. 

In this paper, we show that decreasing a player's computation cost (more generally, locally improving a player's utility) can lead him to a worse global outcome in equilibrium (When more than one equilibrium exists, we compare the equilibrium with the worst outcome for the player). The player actually lose some advantages he had from being weak.
We also discuss some conditions under which this can not happen. These conditions correspond to a game with strict competition.

In the last part of the paper we show how this unintuitive phenomenon, that at first might seem like a problematic aspect of the Nash equilibrium solution concept, can actually explain real life phenomena. We first show that this can explain free riding, where a group lets a weak member of it, that does not contribute for the group's effort, receive credit for the group's success. Moreover the weak player has no incentive to improve and become stronger. We then take an extra step and show that this advantage of weak players might cause people to give signals indicating, that they are stupider than they really are. 
\newline
\newline
{\bf Paper Outline}  The rest of the paper is organized as follows. Section 2 gives a simple examples of the phenomenon, and how computational cost plays a role in it. Section 3 looks at the special case of constant-sum games. Sections 4 and 5 discuss some real-life behaviors and how they can be explained by this framework. We conclude in Section 6 with some discussion.

\section {I dare you to factor}
Consider the following game $G$: \newline
\begin{center}
\begin{tabular}{|c|c|c|}\hline  & factor & don't factor \\\hline factor & 1,1  & 1,3 \\\hline don't factor & 3,1 & -10,-10 \\\hline \end{tabular}
\end{center}

This is an instance of the ``chicken" game, where both players are presented with one large number to factor. A player who factors the number gets a reward of $1$. If one player factors the number and the other does not, then the player who does not factor gets $3$. However if neither player factors the number, they are both punished and need to pay $10$. This game has two pure-strategy equilibria in which one player factors and the other doesn't factor, and one mixed-strategy equilibrium where they both factor with probability $\frac{11}{13}$ and get an expected reward of $1$.

Now consider the following : The year is 2040, Player 1 has a powerful state-of-the-art classical computer, while player 2 has the newest ``Ox" quantum computer, capable of factoring very large numbers efficiently. Both players have a complexity cost associated with every action they take that is represented as a complexity function. Player 1 has a complexity function $c_1$, where not factoring cost nothing, and factoring is not possible, so its complexity is $\infty$. Player 2's complexity function is 0 for both actions. A player's utility is simply the reward of the player minus his complexity cost. This game has only one equilibrium: player 1 does not factor, while player 2 factors. The utility vector in this equilibrium is $(3,1)$.

Now what happens if we change player $1$'s complexity function by giving him an ``Ox" computer? His utilities have obviously improved everywhere, but does this help him or hurt him? The new game we get is identical to the original game without complexity costs, and so has two more equilibria . In both new equilibria  player $1$'s utility is only $1$ instead of $3$. The third equilibrium is identical to the equilibrium with the old costs. This change made things worse for player $1$, since in the worst case his utility with the new function is lower than the utility with the old function. What actually happens here is that when player $1$ has only a classical computer, he has a credible ``threat": he is not going to factor no matter what player $2$ does (This is similar to removing the stirring wheel from the car).  Thus player $2$ must factor. When they both have the ``Ox" computer, that threat is gone and player $2$ is not going to agree to always factor. If offered a free ``Ox'' computer, player $1$ will actually refuse to get it.

The next example shows that the player can do strictly worse in all equilibria by this kind of change to the utilities, not only in the worst case equilibrium. Consider the following game:
\begin{center}
\begin{tabular}{|c|c|c|}\hline  & $a_2$ & $b_2$  \\\hline $a_1$ & 2,1  & -2,2 \\\hline $b_1$ & 3,1 & -1,-1 \\\hline \end{tabular}
\end{center}

In this game, player $1$ has a dominant strategy $b_1$, which leads to only one equilibrium, ($b_1$,$a_2$), with utilities $(3,1)$. Now if player $1$ gets a subsidy of 2 when playing $a_1$,  we get a different equilibrium, ($a_1$,$b_2$), with utilities of $(0,2)$. Note that even the social welfare is worse in this scenario. This change happens because player's $1$ dominant strategy changed from $b_1$ to $a_1$. When player $1$ plays $b_1$, player $2$ prefers to play $a_2$. The change in utility for player $1$ changes the dynamics between the two players, which makes player $2$ also change his actions, and leads to a new equilibrium. Getting the subsidy, which improved player $1$ utilities locally, leads to an equilibrium where his utility is lower. What actually happens is that when player $1$ gets no subsidy for playing  $a_1$, player $2$ knows he can't make player $1$ play anything that is not $b_1$, so he has no choice but to play $a_2$. When he does get the subsidy for $a_1$, player $2$ knows player $1$ will play $a_1$ always so he can play $b_2$. Player $1$ can't threaten player $2$ with playing $b_1$ any more. 

These two examples show that being bounded can help a player, even if it is not just to using a finite automaton. 
They show that changing a player's utilities so that he is better off in any strategy profile (improving his complexity function for every action) might result in an equilibrium in which the player's utility is lower than with his old utilities. This happens because having a bad utility for some profiles gives a player a threat against the other players. This threat is lost when his utility gets better, and some actions that were once unacceptable by him might now be a best response for him to the other players' choice of actions. This is used by the other players to change the equilibrium of the game and create a final result where the player might have lower utility. So, although the player is better off for any strategy profile chosen, the strategy profile that is an equilibrium in the new game is worse for him.

This section showed that, in general games, increasing a player's utility locally (or reducing his complexity cost) can result in an equilibrium where his utility is lower. In a sense, we showed that in some games players actually prefer to be bounded or weak. The next section will look at the special case of constant-sum games.

\section {Constant-sum games}

Constant-sum games have some unique characteristics. In particular they are totally competitive. No one can gain from cooperation. Our intuition is that in these kind of games, a player can not get hurt by improving his utility function. In this section, we show that this intuition is correct and what kind of changes can we make to such games and still have the same effect.

Constant-sum games have the very nice property that by the minimax theorem we know exactly how the players play at equilibrium. In particular, we know that in equilibrium each player plays his defense strategy (his maxmin strategy) which means he plays the strategy that maximizes his minimum payoff - the strategy that gives him the maximum payoff against any strategy the other player plays. We use this fact to show the following theorem.

\begin{theorem} 
Let G be a 2-player constant-sum game with utility functions $u$. Let G' be a game with the same action space as G but with utility functions $u'$ such that for all $\overrightarrow{a}, u_i(\overrightarrow{a})\leq u_i'(\overrightarrow{a})$, and $u'_{-i}$ changed arbitrary. In equilibrium, player $i$'s utility in G' can't be lower than in G.
\end{theorem}

\begin{proof}
Without loss of generality, assume that player $1$'s utility has improved. Now lets look at any equilibrium in G'. If player $2$ plays the same strategy he plays in the equilibrium in G, then if player $1$ plays the same strategy he plays in G, we know that his utility improved by definition. If he plays another strategy then by the definition of equilibrium he gets at least as much as with his strategy in G, otherwise he would want to switch. So in this case player $1$'s utility can't be lower.

If player $2$ plays a different strategy than in G, than we know that if player $1$ plays his strategy in G, he gets at least the same payoff. That is because we know that  G is constant-sum, so player $2$ minimizes player $1$ payoff with that action in G, so if he now changes action player $1$ could only improve. Using the same argument as in the previous case, we know that if player $1$ plays another action, then he must get more than in G.
\end{proof}

This shows that when starting from a constant-sum game, any change to the game that improves one player's utility for every action profile can't hurt him, no matter what changes are done to the other player. The next theorem shows that even games that are not exactly constant-sum but are close to them, have the same characteristics. 

The games we consider are games with utility of the form $u_i(\overrightarrow{a})=u^u_i(\overrightarrow{a})-c_i(a_i)$, where $u^u_i$ is the utility player $i$ gets if there was no cost involved, and we assume the sum of $u^u_i$ over all players is constant for any action profile. With every action $a$, player $i$ has an associated cost (or subsidy) $c_i(a)$, and the player does not gain from other players' costs. These games are the same as the games studied by Ben-Sason, Kalai and Kalai \cite{ben2007approach}. We show that if only one of the players has a cost for his actions (or the other player has a constant cost, which is just a constant-sum game with a different constant) then that player can not lose from changes that improve his utility.

\begin{theorem} 
Let G be a 2-player constant-sum game, with utility functions  $\overrightarrow{u^u}$. In a game $G^c$ in which the utility functions $\overrightarrow{u}$ are of the form $u_i(\overrightarrow{a})=u^u_i(\overrightarrow{a})-c_i(a_i)$, and for one of the players $c_i$ is constant, the other player can not lose in equilibrium from a local improvement of its cost function.
\end{theorem}

\begin{proof}
First assume without loss of generality that player 2's cost function is constant with a value of c, and player 1's cost improved. We use the same idea as Ben-Sason, Kalai and Kalai \cite{ben2007approach}.  Given a game $G^c$, we build a game $H$, that differs from $G^c$ only in the utilities: we define $u^h_i(\overrightarrow{a})= u_i(\overrightarrow{a})+c_{-i}(a_{-i})= u^u_i(\overrightarrow{a})-c_i(a_i)+c_{-i}(a_{-i})$ for $i=1,2$. It easy to verify that any advantage a player can get from switching strategies in $H$, he can get from switching strategies in $G^c$. This means that any equilibrium in $H$ is an equilibrium in $G^c$, and vice versa. Let $\sigma_1,\sigma_2$ be the strategies at an equilibrium that is worst  for player $1$ in $H$, and let $P_{\sigma_i}(a)$ be the probability of playing action $a$ when using strategy $\sigma_i$. Then:\newline
\begin{align*} 
u_1(\sigma_1,\sigma_2) &=\sum_{a,b}p_{\sigma_1}(a)p_{\sigma_2}(b)u_1(a,b)\\
u^h_1(\sigma_1,\sigma_2) &=\sum_{a,b}p_{\sigma_1}(a)p_{\sigma_2}(b)u^h_1(a,b)\\
&=\sum_{a,b}p_{\sigma_1}(a)p_{\sigma_2}(b)u_1(a,b) +\sum_{a,b}p_{\sigma_1}(a)p_{\sigma_2}(b)c_2(b)\\
&=\sum_{a,b}p_{\sigma_1}(a)p_{\sigma_2}(b)u_1(a,b) +c  \text{ ($c_2$ is constant)}\\
&=u_1(\sigma_1,\sigma_2)+c
\end{align*}
This shows that player 1's utility in any equilibrium in $H$ is his utility in the same equilibrium in $G^c$ plus some constant. Moreover when comparing any two equilibria in $H$, the difference in the utility of player $1$ in them is exactly the difference in his utility in $G^c$. This means that $\sigma_1,\sigma_2$ is also the worst equilibrium for player $1$ in $G^c$.

$H$ is a game of the type described in Theorem $1$. So, by that theorem, if we change player $1$'s cost function $c_1$ to a cost function $c_1'$, where he pays no more than with $c_1$ for his actions, player $1$ gets at least the same utility in the worst case. We call the games created by changing $c_1$ to $c_1'$ $H'$ and $G^{c'}$. By the same argument as before, the new worst-case equilibrium for player $1$ in $H'$ is also the worst-case equilibrium for player $1$ in $G^{c'}$, and the difference between the utility of player $1$ in the worst-case equilibria for player $1$ of $H$ and $H'$ is exactly equal to the differences between player $1$'s utility  in the worst-case equilibria for player $1$ of $G^c$ and $G^{c'}$. This means that in $G^{c'}$ he is also at least as well off as he was in $G^c$.
\end{proof}

The next example shows that if both players have non-constant costs this property fails to hold. Consider the following game:

\begin{center}
\begin{tabular}{|c|c|c|}\hline $u^u$ & $a_2$ & $b_2$  \\\hline $a_1$ & 6,0  & 2,4 \\\hline $b_1$ & 4,2 & 1,5 \\\hline \end{tabular}
\end{center}

where the costs are $0$ for all actions for all players. The game described has only one equilibrium: $(a_1,b_2)$ with utility profile $(2,4)$. Now consider the situation where $c_1(a_1)=-1.5$, $ c_1(b_1)=0$, $c_2(a_2)=0$, $c_2(b_2)=-3.5$. This game has only one equilibrium: both players play each action with probability 0.5. The expected utility is (2.5,1). By changing only $c_1$ to $c_1'$, which is 0 for both actions, which obviously locally improve the players utility, the equilibrium is changed back to $(a_1,b_2)$ with utilities $(2,0.5)$, and thus the player is globally worse than before.

This section showed how in constant-sum games, a player can always gain from a local improvement in his utility, and that it is also true in games which are close to constant sum games (a change to only one player makes them constant-sum again). There are of course other changes that can not harm a player, even in games which are not constant sum. For example, improving both actions in the same amount (which is just like giving free money no matter what the player does), or improving the cost of an already dominant strategy by more than that of other strategies. \newline

\section {I would rather stay stupid}

The phenomenon of free riding is well observed and studied. One flavor of it occurs when a part of a group gets credit for the work of others without contributing anything. 
We argue that this can be explained by the weak player's advantage we described, and moreover that it also explains why there is a negative incentive for weak players to improve.
We illustrate this by an example that shows how a weak player gets credit without any contribution, and why the rest of the group might agree to it.

Consider a scenario where two students are assigned to work on an assignment together, but they can't meet and exchange work before its deadline. The assignment has $10$ questions. In order to solve each question, you must first have the answer to the previous question. Both students gain one point for every question any one of them solves; if they both solve the same question they still get only $1$ each. This means that their utility can be written as $u=max(x_1,x_2)$  where $x_i$ is the number of questions student $i$ solves. The first student has a complexity cost of $0.1$ for every question he solves, while the second student has a complexity cost of $0.1$ for the first $7$ question, and a cost of $1.1$ for the rest. This game has only one equilibrium: the first student does all $10$ questions, while the second student free rides (does nothing) and gets the credit. This is because for every question the second student will solve after the 7th he will get $-0.1$ utility, so he will not do more then $7$ questions, while the first student prefers solving 10 question by himself to doing nothing and having the second student solve only $7$ questions. The utility for the second student in this game is $10$.

Now what would have happened if the second student were ``smarter", and had the same complexity cost as the first student? In this situation, the game would have one more pure strategy equilibrium (it also has a mixed strategy equilibrium), in which the second student solves all $10$ questions and the first student does nothing. In this equilibrium, the second student's utility is $9$, which is lower than before. This shows that the second student has no incentive to try to get ``smarter". By free riding, he ensures himself the highest utility possible in this game, while the other student has no choice but to do all the work. This happens because the second student has a credible threat: no matter what the first student does, he won't solve more than $7$ questions.
This shows that the second student has no incentive to improve. If he gets smarter he actually loses the threat and gets a lower payoff.

\section {I am better than I seem}

In this section, we explore why people sometimes pretend to be weak, when they are really not, which is also a well observed phenomenon. 
Intuitively, we show that the reason for this is that acting weak, lowers the expectation from the players, and allows them to invest less effort. 
To do that we do not look at the added value a player gets from better utility, but instead look at a scenario where a player would rather behave as if he had lower utility, since in equilibrium it gets him a higher payoff. 

We use the spirit of Spence's \cite{spence1973job} signaling model, which is traditionally used to show how players signal how good they are to the other players, to instead show that players might sometimes want to do the opposite, and signal that they are even worse then they really are. Spence shows how education can be seen  as a signal in the hiring market, which helps employers decides the wages to offer job candidates. He defines an information-feedback cycle, consisting of the employer's beliefs, the offered wages as a function of the signals, the signal chosen by applicants, and the final observation of the hired employees by the employers (which feeds back into their beliefs). He defines an equilibrium in this model as a situation where the beliefs of the employer are self confirming, that is, do not change as a result of the final observation of the employers that were hired based on the previous beliefs. 

We use a variant of this model to show that people might even try to seem stupider than they really are (or more generally signal that their utility\textbackslash complexity functions are weaker) to get the power of a credible threat. Consider an educational institution that wants to divide its students into two classes, regular and honors. For every student that is placed in the honors class and is able to pass it, the school gets a utility of $1$, but if the student is placed in the honors class and fails, the school get a utility of $-1$. For every student that is placed in the regular class, the school gets a utility of $0$. A student has two options: he can either relax and easily pass the regular class but fail the honors class, or he can work hard, in which case he passes both classes. If he passes a class he gets utility of 100, and if he fails he gets utility of 0.

There are three types of students. A slow student, who has a cost of $100$ for working hard, a moderate student, who has a cost of $7$, and a fast student, who has a cost of $3$. The cost of relaxing is $0$ for all students. The school would like to place the slow students in the regular class and the moderate and fast students in the honors class, but it cannot tell what student fall into each category. 

To help it with the process, the school decides to do a preliminary placement test for students. The test has $10$ questions, and the school decides that a student who solves $7$ or more questions will be placed in the honors class. To motivate the students to perform well, the school offers them the option of skipping one hour of class without being punished for every question they solve. Skipping an hour has a utility of $1$. The three types of students have different costs for this test. All students have a cost of $0$ for the first six questions. For every question after that, the slow student has a cost of $1.1$, the moderate student has a cost of $0.5$, and the fast student has a cost of $0.2$. The school is unaware of these exact costs (as in the Spence model, where the employer is unaware of the exact education costs of the different groups), but designs the test knowing that both the fast and moderate students have a positive incentive to answer all the questions, while the slow student will not answer more than six. 

It is easy to see that in order to maximize their utility, the slow and moderate students will answer only six questions (giving them utility 106), and the fast students will answer all ten questions (giving them utility 106.2). Doing badly in the exam acts as a signal of being slower (and is negatively correlated with the utility, as required by the Spence model). Since the only way for the school to figure out if it did the right thing is to see if someone failed the honors class (this is slightly different from Spence's original model, where the employer gets complete feedback), and no students in the honors class will fail, the school's beliefs are self-confirming, so this gives an equilibrium. 

As in our previous example, the slow students have no motivation to become moderate, thus changing from the cost of the slow student to that of the moderate student has value of of $0$. The value for both the moderate and the slow student of getting the complexity of the fast player is positive in this example. 

This example shows that people will sometimes prefer to be considered less smart then they really are, in order to take advantage of the threat of having higher complexity. 

\section {Conclusion}

This paper gives a negative answer to our motivating question - it shows that in some scenarios players actually prefer to be bounded. 
They can actually gain an advantage from having higher complexity costs.
More generally, the paper shows that locally improving a player's utility for every action profile, can be bad for him in equilibrium.

We then use this counter-intuitive fact, that might seem as just a minor flaw in the solution concept, to actually explain familiar human behaviors that may seem at first to be irrational.
It will be interesting to find other observed phenomena that might be explained by these ideas.

The paper also characterize games for which this can't happen. We show that strict competition actually prevent such behaviors.
This suggests that the cooperative nature of some tasks might cause people to invest less, act as if they are weak and free ride the other participants in a task.
This can be used by mechanism designers and multi agent systems designers when deciding on the incentives structure they want to use, and on how to encourage corporation in their environment.


\bibliographystyle{abbrv}
\bibliography{bibliography}  

\begin{thebibliography}{10}

\bibitem{porath1993repeated}
E.~Ben-Porath.
\newblock Repeated games with finite automata.
\newblock {\em Journal of Economic Theory}, 59:17--32, 1993.

\bibitem{ben2007approach}
E.~Ben-Sasson, A.~Kalai, and E.~Kalai.
\newblock An approach to bounded rationality.
\newblock In {\em Advances in Neural Information Processing Systems 19:
  Proceedings of the 2006 Conference}, pages 145--153. The MIT Press, 2007.

\bibitem{blackwell1951comparison}
D.~Blackwell.
\newblock Comparison of experiments.
\newblock In {\em Second Berkeley Symposium on Mathematical Statistics and
  Probability}, volume~1, pages 93--102, 1951.

\bibitem{blackwell1953equivalent}
D.~Blackwell.
\newblock Equivalent comparisons of experiments.
\newblock {\em The Annals of Mathematical Statistics}, pages 265--272, 1953.

\bibitem{gilboa1989bounded}
I.~Gilboa and D.~Samet.
\newblock Bounded versus unbounded rationality: The tyranny of the weak.
\newblock {\em Games and Economic Behavior}, 1(3):213--221, 1989.

\bibitem{HP10}
J.~Y. Halpern and R.~Pass.
\newblock Game theory with costly computation.
\newblock In {\em Proc.~First Symposium on Innovations in Computer Science},
  2010.

\bibitem{hirshleifer1971private}
J.~Hirshleifer.
\newblock The private and social value of information and the reward to
  inventive activity.
\newblock {\em The American Economic Review}, 61(4):561--574, 1971.

\bibitem{kamien1990value}
M.~Kamien, Y.~Tauman, and S.~Zamir.
\newblock On the value of information in a strategic conflict.
\newblock {\em Games and Economic Behavior}, 2(2):129--153, 1990.

\bibitem{Neyman1985}
A.~Neyman.
\newblock Bounded complexity justifies cooperation in the finitely repeated
  prisoners' dilemma.
\newblock {\em Economics Letters}, 19(3):227--229, 1985.

\bibitem{neyman1991positive}
A.~Neyman.
\newblock The positive value of information.
\newblock {\em Games and Economic Behavior}, 3(3):350--355, 1991.

\bibitem{rubinstein1986finite}
A.~Rubinstein.
\newblock Finite automata play the repeated prisoner's dilemma.
\newblock {\em Journal of Economic Theory}, 39(1):83--96, 1986.

\bibitem{spence1973job}
M.~Spence.
\newblock Job market signaling.
\newblock {\em The Quarterly Journal of Economics}, 87(3):355, 1973.

\end{thebibliography}

\end{document}